\documentclass[
final
]{dmtcs-episciences}

\usepackage{amsmath,amssymb,amsfonts,amsthm}
\usepackage{array,setspace,multicol,multirow}
\usepackage{hyperref,enumerate,verbatim,graphicx}
\usepackage{url}
\usepackage[dvipsnames]{xcolor}
\usepackage{anyfontsize}
\usepackage{xspace}
\usepackage{cleveref}

\DeclareSymbolFont{rsfscript}{OMS}{rsfs}{m}{n}
\DeclareSymbolFontAlphabet{\mathrsfs}{rsfscript}
\DeclareMathOperator{\rt}{rt}
\DeclareMathOperator{\lspan}{span}
\newcommand{\botequiv}{\equiv}
\newcommand{\scalar}{\odot}

\newtheorem{theorem}{Theorem}
\newtheorem{corollary}[theorem]{Corollary}
\newtheorem{lemma}[theorem]{Lemma}
\newtheorem{proposition}[theorem]{Proposition}
\newtheorem{conjecture}[theorem]{Conjecture}

\numberwithin{theorem}{section}
\newtheorem{definition}{Definition}

\numberwithin{definition}{section}

\newcommand{\PSPACE}{\textnormal{\textsf{PSPACE}}\xspace}

\makeatletter
\def\Ddots{\mathinner{\mkern1mu\raise\p@
\vbox{\kern7\p@\hbox{.}}\mkern2mu
\raise4\p@\hbox{.}\mkern2mu\raise7\p@\hbox{.}\mkern1mu}}
\makeatother

\renewcommand{\O}{\mathcal{O}}

\author[Berlinkov, M.V., Ferens, R., Ryzhikov A., and Szyku{\l}a, M.]{%
Mikhail~V.~Berlinkov\affiliationmark{1} \and
Robert~Ferens\affiliationmark{2} \and
Andrew~Ryzhikov\affiliationmark{3}\thanks{Supported by the National Science Centre, Poland under project number 2022/46/E/ST6/00230.} \and
Marek~Szyku{\l}a\affiliationmark{2}\thanks{Supported by the National Science Centre, Poland under project number 2021/41/B/ST6/03691.}
}
\title[Synchronization of partial DFAs]{Synchronization of strongly connected\\partial DFAs and prefix codes}
\affiliation{%
Independent Researcher, Canada\\
University of Wroc{\l}aw, Wroc{\l}aw, Poland\\
University of Warsaw, Warsaw, Poland}
\begin{document}
\publicationdata{vol. 28:2}{2026}{34}{10.46298/dmtcs.16465}{2025-09-04; 2025-09-04; 2026-05-27}{2026-06-01}
\maketitle
\begin{abstract}
We study synchronizing partial DFAs, which extend the classical concept of synchronizing complete DFAs and are a special case of synchronizing unambiguous NFAs. A partial DFA is called synchronizing if it has a word (called a \emph{reset word}) whose action brings a non-empty subset of states to a unique state and is undefined for all other states.
The class of strongly connected partial DFAs is precisely the class of DFAs recognizing the Kleene star of prefix codes.
While in the general case the problem of checking whether a partial DFA is synchronizing is \PSPACE-complete, we show that in the strongly connected case, this problem can be efficiently reduced to the same problem for a complete DFA. Using combinatorial, algebraic, and formal languages methods, we develop techniques that relate main synchronization problems for strongly connected partial DFAs to the same problems for complete DFAs. In particular, this includes the \v{C}ern\'{y} and the rank conjectures, the problem of finding a reset word, and upper bounds on the length of the shortest reset words of literal automata of finite prefix codes. We conclude that solving fundamental synchronization problems is equally hard in both models, as an essential improvement of the results for one model implies an improvement for the other.
\end{abstract}
\section{Introduction}

Synchronization is an important concept in various domains of computer science that consists in regaining control over a system by applying (or observing) a specific set of input instructions. These instructions are usually required to lead the system to a fixed state no matter in which state it was at the beginning. This idea has been studied for automata (deterministic \cite{Cerny1964,Volkov2022Survey}, nondeterministic \cite{Imreh1999}, unambiguous \cite{Beal2008}, weighted and timed \cite{Doyen2014WeightedTimedAutomata}, partially observable~\cite{Larsen2014}, register \cite{Babari2016}, nested word~\cite{Chistikov2019}), parts orienting in manufacturing~\cite{Ep1990,Na1986}, testing of reactive systems~\cite{Sandberg2005Survey}, variable length codes~\cite{BPR2010CodesAndAutomata}, and Markov Decision Processes~\cite{Doyen2014,Doyen2019}.

In this paper, we study the synchronization of partial DFAs, which are a generalization of complete DFAs and a special case of unambiguous NFAs.
We are motivated by applications of this model and its connections with others, as well as the need for new techniques applied to partial DFAs.
The problems for strongly connected partial DFAs are a motivation for further development and generalization of the methods applied for complete DFAs, since, as we show, these models are closely related.
We also hope that our methods will serve as a step toward studying a wider class of strongly connected unambiguous~NFAs.

\subsection{Observing a reactive system}

Consider a finite-state reactive system modeled by a partial DFA (by partial we mean that for some states there can be no outgoing transitions corresponding to some letters).
The observer knows the structure of the DFA but does not know its current state.
At every step, the DFA reads a letter (also known to the observer) and transits to another state.
The observer wants to eventually learn the actual state of the DFA. Since the DFA is deterministic, once a state is known, it will be known forever.

In this setting, the actual state is known if and only if the system reads a \emph{reset word} -- a word that transits a non-empty set of states to a single state and is undefined for all other states. The presence of undefined transitions indicates that certain actions cannot be performed from certain states, which can be essential for synchronization.

For several identical systems running in parallel and receiving the same input (but possibly starting from different states), the presence of a reset word in the input guarantees that all systems end up in the same state.
This idea can be used in robotics, where a sequence of passive obstacles is used for orienting a large number of arbitrarily rotated parts arriving simultaneously on a conveyor belt (\cite{Ep1990,Na1986}, see also \cite{Volkov2008Survey} for an illustrative example).

Reactive systems (such as Web servers, communication protocols, operating systems and processors) are systems developed to run without termination and interact through visible events, so it is natural to assume that the system can return to any state from any other state (NFAs with this property are called \emph{strongly connected}).
The probabilistic version of the described problem for strongly connected partial DFAs has been considered in the context of $\varepsilon$-machines \cite{TrCr11Exact}. In particular, the observer knows the state of an $\varepsilon$-machine precisely if and only if a reset word for the underlying partial DFA was applied.
Some experimental results on finding shortest reset words for partial DFAs were recently presented in~\cite{Shabana2019}.

\subsection{Synchronizing automata}

There exist several definitions that generalize the notion of a synchronizing complete DFA to larger classes of NFAs. In this subsection, we describe the notion which preserves most of the properties of the complete DFAs case, and in~\Cref{subs:car} we briefly describe alternative notions.

An NFA is called \emph{unambiguous} if for every two states $p, q$ and every word $w$, there is at most one path from $p$ to $q$ labeled by $w$ \cite{Beal2008}.
In the strongly connected case, this is equivalent to a more classical definition of an unambiguous NFA with chosen initial and final states, if there is a unique initial state and a unique final state.
An unambiguous NFA is called \emph{synchronizing} if there exist two non-empty subsets $C, R$ of its states and a word $w$ (called a \emph{reset} word) such that its action maps every state in $C$ exactly to the whole set $R$, and is undefined for all states outside $C$ \cite{Beal2008}.
For partial DFAs, the set $R$ has size one~\cite{BPR2010CodesAndAutomata}, and, for complete DFAs, the set $C$ is also the whole set of states~\cite{Volkov2022Survey}.

Partial DFAs are thus a natural intermediate class between unambiguous NFAs and complete DFAs.
The bounds on the length of shortest reset words in strongly connected partial DFAs have not been studied before.
The famous \v{C}ern\'{y} conjecture, which is one of the most longstanding open problems in automata theory, states that for an $n$-state complete DFA we can always find a reset word of length at most $(n-1)^2$, unless there are no reset words. The best known upper bound is cubic in~$n$~\cite{Shitov2019,Szykula2018ImprovingTheUpperBound}, and the problem of deciding whether a complete DFA is synchronizing is solvable in time quadratic in~$n$ \cite{Volkov2022Survey}.
For an $n$-state strongly connected unambiguous NFA, the best known upper bound on the length of a shortest reset word is $n^5$, and the existence of a reset word is verifiable in polynomial time~\cite{Ryzhikov2019WORDS}.
The same upper bound holds for the length of the shortest mortal words in strongly connected unambiguous NFAs~\cite{KieferMascle2019}, whereas partial DFAs admit a tight quadratic bound~\cite{Rystsov1983PolynomialCompleteProblems}.

\subsection{Synchronizing codes}\label{subs:sync-codes}

A \emph{variable-length code} $X$ (which we call a \emph{code}) is a set of finite words over a finite alphabet $\Sigma$, such that no word over $\Sigma$ can be written as a concatenation of codewords of $X$ in two different ways.
Such codes (especially Huffman codes \cite{Huffman1952}) are widely used for lossless data compression.
Since the lengths of codewords can be different, one transmission error can spoil the whole decoding process, causing a major data loss.
Also, for general codes, decoding a part of a message (e.g., a segment of a compressed video stream) is not possible without decoding the whole message.

These issues can be addressed by using synchronizing codes. A code $X$ is called \emph{synchronizing} if there exists a \emph{synchronizing} word $w \in X^*$ such that for every $uwv \in X^*$ we have $uw, wv \in X^*$.
The occurrence of the word $ww$ thus stops error propagation and allows parallel decoding of the two parts of the message.
More generally, each appearance of the word $ww$ in a coded message allows to run the decoding independently from the position after the first $w$.

A code is called \emph{prefix} if none of its codewords is a prefix of another codeword. Such codes allow to obtain the correct partition of a message into codewords one by one by going from left to right.
Even if a code is synchronizing, there are no guarantees that a synchronizing word will appear in a message. Codes where every long enough concatenation of codewords is synchronizing are called \emph{uniformly synchronizing} \cite{BPR2010CodesAndAutomata,Bruyere1998}.
A prefix code is called \emph{maximal} if it is not a subset of another prefix code.
All non-trivial uniformly synchronizing finite prefix codes are non-maximal \cite{BPR2010CodesAndAutomata}.

\subsection{Automata for the Kleene star of codes} \label{subs:stars} 

A code recognized by an NFA as a language is called \emph{recognizable}.
In particular, every finite code is recognizable.
To argue about synchronization properties of a recognizable code $X$, special NFAs recognizing~$X^*$ are studied.
These NFAs have a unique initial and final state $r$ such that the set of words labeling paths from $r$ to itself coincides with $X^*$, thus they are also strongly connected.
Provided a recognizable code $X$, an NFA with the described properties can be chosen to be unambiguous \cite{BPR2010CodesAndAutomata}. Moreover, this NFA can be chosen to be a partial (respectively, a complete) DFA if and only if $X$ is a recognizable prefix (respectively, recognizable maximal prefix) code \cite{BPR2010CodesAndAutomata}.

For such an unambiguous NFA with the properties as above, $X$ is synchronizing if and only if the NFA is synchronizing, and the length of a shortest synchronizing word for $X$ is at most the length of a shortest reset word of the NFA plus twice its number of states \cite[Chapter~4]{BPR2010CodesAndAutomata}.

Finite prefix codes admit a direct construction of partial DFAs with the described properties, called literal (or prefix) automata. Let $X$ be a finite prefix code over an alphabet $\Sigma$.
The \emph{literal automaton} $\mathrsfs{A}_X = (Q,\Sigma,\delta)$ is constructed as follows.
The set of states $Q$ is the set of all proper prefixes of the words in $X$, and the transition function is defined as follows: $\delta(q,x) = qx$ if $qx \notin X$ and $qx$ is a proper prefix of a word in $X$, $\delta(q,x) = \varepsilon$ if $qx \in X$, and $\delta(q,x) = \bot$ otherwise.
The state corresponding to the empty prefix $\varepsilon$ is called the \emph{root state}.
The \emph{height} of a literal automaton is the length of a longest path of its transitions without repetition of states; equivalently, this is the length of the longest word in~$X$ minus one.
Note that 
the number of states of $\mathrsfs{A}_X$ is at most the total length of all codewords of $X$, which allows to directly transfer upper bounds from literal automata to finite prefix codes.
An example of a literal automaton is shown in~\Cref{fig:examples} (right).
The literal automaton of a prefix code can be used as a decoder for this code by adding output labels to the transitions \cite{BPR2010CodesAndAutomata}.

\subsection{Carefully synchronizing DFAs}\label{subs:car}

For general NFAs, synchronizability can be generalized to D$i$-directability for $i = 1, 2, 3$ \cite{Imreh1999}. As discussed in \cite[Section~6.3]{Vorel2016}, for partial DFAs the notions of D$1$- and D$3$-directing words both coincide with \emph{carefully synchronizing} words. These are words sending every state of a partial DFA to the same state, without using any undefined transitions. A D$2$-directing word for a partial DFA is either carefully synchronizing or mortal (undefined for every state). The definitions of carefully synchronizing and D$2$-directing words are different from our definition of synchronizing words for partial DFAs.

A carefully synchronizing word can be applied to a partial DFA at any moment without the risk of using an undefined transition. This comes at a high cost: even for strongly connected partial DFAs, the shortest carefully synchronizing words can have exponential length \cite[Proposition~9]{Vorel2016}, and the problem of checking the existence of such a word is \PSPACE-complete \cite[Theorem~12]{Vorel2016}, in contrast with the case of complete DFAs.
On the contrary, the notion of a synchronizing partial DFA preserves most of the properties of a synchronizing complete DFA, at least in the strongly connected case. Note that every carefully synchronizing word is synchronizing, but the converse is not true.

While for complete DFAs the property of being strongly connected is not essential for many synchronization properties \cite{Volkov2022Survey}, the situation changes dramatically for partial DFAs.
Partial DFAs that are not strongly connected can have exponentially long shortest reset words, and the problem of checking the existence of a reset word is \PSPACE-complete \cite{Berlinkov2014OnTwoAlgorithmicProblems}.
Thus, strong connectivity is indeed necessary to obtain good bounds and algorithms. As explained above, for reactive systems and prefix codes this requirement comes naturally.

\subsection{Our contribution and organization of the paper}

We prove a number of results for strongly connected partial DFAs connected with the \v{C}ern\'{y} conjecture and its generalizations.
Where possible, we suggest methods that allow to relate the partial case with the complete case, instead of directly reproving known results in this more general setting.
In this way, we do not have to go into the existing proofs, and future findings concerning the complete case should be often immediately transferable to the partial case.

We start from basic properties and introduce more advanced techniques along with their applications.
First, we investigate the rank conjecture, which is a generalization of the \v{C}ern\'{y} conjecture from the case of synchronizing automata to the case of all automata. We show that the rank conjecture for complete DFAs implies an analogous statement for partial DFAs (\Cref{thm:rank_conjecture_transfer}).
For this, we introduce our first basic tool called a \emph{fixing automaton}, which is a complete DFA obtained from a partial one and sharing some of its properties.
Our result shows a general way for transferring upper bounds from the case of complete DFAs to partial DFAs, e.g., we immediately get that the rank conjecture holds true for partial Eulerian automata (\Cref{cor:eulerian}).

To connect the \v{C}ern\'{y} conjecture for the cases of complete and partial DFAs, we need more involved techniques, since the construction developed for the rank conjecture does not preserve the property of being synchronizing. We introduce a \emph{collecting automaton}, which extends the concept of the fixing automaton. We use it to show that all upper bounds on the length of the shortest reset words, up to a subquadratic additive component (linear in the case of the \v{C}ern\'{y} bound), are equivalent for partial and complete DFAs (\Cref{thm:cerny_transfer}).
We also use it to prove that the problems of deciding synchronizability and finding a reset word of a strongly connected partial DFA can be effectively reduced to the same problems for a complete DFA (\Cref{sec:alg_issues}).
This also means that possible improvements of the complexity of the best-known algorithms for these problems for complete DFAs should directly apply to partial DFAs.

As discussed in \Cref{subs:sync-codes} and \Cref{subs:stars}, one of the main motivations for studying synchronization of strongly connected partial DFAs is a direct correspondence with synchronization of recognizable prefix codes. An important special case is when the prefix code is finite. We investigate it by studying literal automata of finite prefix codes and obtain stronger upper bounds than those for the general case of strongly connected partial (or complete) DFAs. We show that the length of the shortest reset words for literal automata of finite prefix codes is at most $\O(n \log^3 n)$, where $n$ is the number of states of the automaton (\Cref{cor:literal_aut}).
This upper bound asymptotically matches the strongest known upper bound for maximal prefix codes (which becomes now a special case), but it is not transferred directly, as key statements do not hold in the same way for non-maximal prefix codes.
To prove it, we first show that the literal automaton of a finite prefix codes admits a word of linear length whose action sends all the states to a non-empty subset of small size (\Cref{thm:log-rank}). It establishes a natural combinatorial property of finite prefix codes and constitutes the most involved proof in this paper. Once we show the existence of such a word, we use one more construction called the \emph{induced automaton}, which is a generalization of linear algebraic techniques to the case of partial DFAs (\Cref{sec:induced}). This particular construction extends the existing techniques originally developed for complete DFAs but simultaneously comes with a new simpler and more general proof.

Finally, we show that the lower bounds for strongly connected partial DFAs are asymptotically the same even if we ensure the existence of undefined transitions (\Cref{sec:lower_bounds}).
In other words, undefined transitions do not help in general, as we cannot significantly improve upper bounds for such automata without doing that for the complete case.

This paper is the full version of a conference paper \cite{BFRS21SynchronizingStronglyConnectedPartialDFAs}.

\section{Preliminaries}

A \emph{partial deterministic finite automaton} $\mathrsfs{A}$ (a partial DFA for short) is a triple $(Q,\Sigma,\delta)$, where $Q$ is a set of \emph{states}, $\Sigma$ is an \emph{input alphabet}, and $\delta$ is partial function $Q \times \Sigma \rightharpoonup Q$ called the \emph{transition function}. Note that the automata we consider do not have any initial or final states.
We extend $\delta$ to a partial function $Q \times \Sigma^* \rightharpoonup Q$ as usual: we set $\delta(q, wa) = \delta(\delta(q,w), a)$ for $w \in \Sigma^*$ and $a \in \Sigma$.
For a state $q \in Q$ and a word $w \in \Sigma^*$, if the action $\delta(q,w)$ is undefined, then we write $\delta(q,w)=\bot$.
Note that if $\delta(q,w) = \bot$ for a word $w \in \Sigma^*$, then $\delta(q,wu) = \bot$ for every word $u \in \Sigma^*$.
A DFA is \emph{complete} if all its transitions are defined, and it is \emph{incomplete} otherwise.
An partial DFA is \emph{strongly connected} if for every two states $p,q \in Q$ there is a word $w \in \Sigma^*$ such that $\delta(p,w) = q$.

By $\Sigma^i$ we denote the set of all words over $\Sigma$ of length exactly $i$, and by $\Sigma^{\le i}$ the set of all words over $\Sigma$ of length at most $i$.
For two sets of words $W_1, W_2 \subseteq \Sigma^*$, by $W_1 W_2$ we denote their product $\{w_1 w_2 \in \Sigma^* \mid w_1 \in W_1, w_2 \in W_2\}$.
The empty word is denoted by~$\varepsilon$.
Throughout the paper, by $n$ we always denote the number of states $|Q|$.

Given $S \subseteq Q$, the \emph{image} of $S$ under the action of $w$ is $\delta(S,w) = \{\delta(q,w) \mid q \in S,\ \delta(q,w) \neq \bot\}$.
The \emph{preimage} of $S$ under the action of $w$ is $\delta^{-1}(S,w) = \{q \in Q \mid \delta(q,w) \in S\}$.
Since $\mathrsfs{A}$ is deterministic, for disjoint subsets $S,T \subseteq Q$, their preimages under the action of every word $w \in \Sigma^*$ are also disjoint.

The \emph{rank} of a word $w$ is the size of the image of $Q$ under the action of this word, i.e., $|\delta(Q,w)|$.
In contrast with complete DFAs, partial DFAs may admit words of rank zero; these words are called \emph{mortal}.
Words of non-zero rank are called \emph{non-mortal}.
A word of rank $1$ is called \emph{reset}, and if the DFA admits such a word then it is called \emph{synchronizing}.
The \emph{reset threshold} $\rt(\mathrsfs{A})$ is the length of the shortest reset words of $\mathrsfs{A}$.

We say that a word $w$ \emph{compresses} a subset $S \subseteq Q$, if $\delta(S,w) \neq \emptyset$ and $|\delta(S,w)| < |S|$.
A subset that admits a compressing word is called \emph{compressible}.
There are two ways to compress a subset $S \subseteq Q$ with $|S| \ge 2$ in a partial DFA.
One possibility is the \emph{pair compression}, which is the same as in the case of a complete DFA, i.e., mapping at least two states $p,q \in S$ to the same state (but not to~$\bot$).
The other possibility is to map at least one state from $S$, but not all states from $S$, to~$\bot$.
Sometimes, a subset can be compressed in both ways simultaneously.
We say that a word $w$ \emph{synchronizes} a set $S \subseteq Q$ if $|\delta(S,w)| = 1$.

\begin{figure}[htb]\centering
\includegraphics{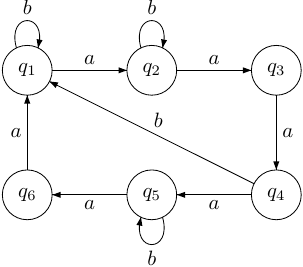}\hspace{2cm}
\includegraphics{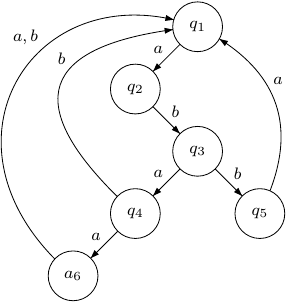}
\caption{\normalsize Left: a strongly connected partial 6-state binary DFA;
right: the literal automaton of the prefix code $\{abaaa,abaab,abab,abba\}$.}
\label{fig:examples}
\end{figure}

An example of a strongly connected partial DFA is shown in~\Cref{fig:examples} (left).
We have two undefined transitions: $\delta(q_3,b)=\delta(q_6,b)=\bot$.
The unique shortest reset word is $bab$: $\delta(Q,b)=\{q_1,q_2,q_5\}$, $\delta(Q,ba)=\{q_2,q_3,q_6\}$, and $\delta(Q,bab) = \{q_2\}$.
However, in contrast with the case of a complete DFA, the preimage $\delta^{-1}(\{q_2\},bab) = \{q_1,q_4\}$ is not $Q$.

\section{Upper bounds}\label{sec:upper_bounds}

\subsection{Inseparability equivalence}

Let $\mathrsfs{A}=(Q,\Sigma,\delta)$ be a partial DFA.
We define the inseparability relation $\botequiv$ on $Q$.
Two states are \emph{separable} if there is a word whose action is defined for exactly one of them.

\begin{definition}
The \emph{inseparability equivalence} $\botequiv$ on $Q$ is defined as follows:
\[ p \botequiv q \quad\text{ if and only if }\quad \forall_{u \in \Sigma^*}\ \left(\delta(p,u) \ne \bot \Leftrightarrow \delta(q,u) \ne \bot\right) .\]
\end{definition}

The same relation is considered in~\cite[Section~1.4]{BPR2010CodesAndAutomata} if all states of the partial DFA are final.
Also, if we replace $\bot$ with a unique final state, then $\botequiv$ is the well-known Myhill-Nerode congruence on words in a complete DFA.
Under a different terminology, it also appears in the context of $\varepsilon$-machines, where non-equivalent states are called \emph{topologically distinct}~\cite{TrCr11Exact}.

For a subset $S \subseteq Q$, let $\kappa(S)$ be the number of equivalence classes that have a non-empty intersection with $S$.
In the partial DFA from~\Cref{fig:examples} (left), we have three equivalence classes, namely, $q_1 \botequiv q_4$, $q_2 \botequiv q_5$, and $q_3 \botequiv q_6$.

Our first auxiliary lemma states that every subset $S \subseteq Q$ with $\kappa(S) \ge 2$ can be compressed by a short word which decreases the number of intersected equivalence classes.
This is done by mapping to $\bot$ all the states of $S$ from at least one equivalence class, but not the whole set $S$.
A linear upper bound can be inferred from a standard analysis of the corresponding Myhill-Nerode congruence, but we will need a more precise bound in terms of $\kappa(S)$.

\begin{lemma}\label{lem:voiding_lemma}
Let $\mathrsfs{A}=(Q,\Sigma,\delta)$ be a partial DFA, and let $S \subseteq Q$ be a subset such that $\kappa(S) \ge 2$.
Then there is a word $w \in \Sigma^*$ of length at most $\kappa(Q)-\kappa(S)+1 \le n-|S|+1$ and such that $1 \le \kappa(\delta(S,w)) < \kappa(S)$.
\end{lemma}
\begin{proof}
We define auxiliary relations on $Q$ that are restricted to words of certain lengths.
For $k \ge 0$, we define:
\[ p \botequiv_k q \quad\text{ if and only if }\quad \forall_{u \in \Sigma^{\le k}}\ \left(\delta(p,u) \ne \bot \Leftrightarrow \delta(q,u) \ne \bot\right) .\]

Clearly, $\botequiv_0$ has all states in one equivalence class, and there is some $m$ such that $\botequiv_m$ is the same as $\botequiv$ because the number of different actions of words is finite.
Also, for every $p,q \in Q$ and $k \ge 0$, if $q \not\botequiv_k p$, then $q \not\botequiv_{k+1} p$.

We show that if for some $k$, $\botequiv_k$ is the same as $\botequiv_{k+1}$, then the chain of relations stabilizes at~$\botequiv_k$, i.e., all relations
$\botequiv_k,\botequiv_{k+1},\botequiv_{k+2},\ldots$
are the same as $\botequiv$.
Assume for a contradiction that $\botequiv_k$ is the same as $\botequiv_{k+1}$, but $\botequiv_{k+2}$ is different from them.
This means that there are two distinct states $p,q \in Q$ such that $p \botequiv_k q$, $p \botequiv_{k+1} q$, and $p \not\botequiv_{k+2} q$.
Hence, there exists a word $u$ of length $k+2$ such that, without loss of generality, $\delta(p,u)=r \ne \bot$ and $\delta(q,u)=\bot$.
Write $u = av$, where $a \in \Sigma$ and $v \in \Sigma^*$.
Then $\delta(\{p,q\},a) = \{p',q'\}$ for some distinct states $p',q' \in Q$.
Since $|v|=k+1$, $\delta(p',v)=r$ and $\delta(q',v)=\bot$, we have $p' \not\botequiv_{k+1} q'$, and from our assumption that $\botequiv_k$ is the same as $\botequiv_{k+1}$, we also have $p' \not\botequiv_k q'$.
This means that there exists a word $v'$ of length at most $k$ such that $\delta(p',v')=r' \ne \bot$ and $\delta(q',v')=\bot$, or vice versa.
But then the action of $av'$ maps exactly one of $p$ and $q$ to $\bot$.
Since $|av'| \le k+1$, this yields a contradiction with $p \botequiv_{k+1} q$.

If for some $k$, $\botequiv_{k+1}$ is different from $\botequiv_k$, then the number of equivalence classes in $\botequiv_{k+1}$ is larger by at least one than the number of equivalence classes in $\botequiv_k$.
Note that the number of equivalence classes is limited by $\kappa(Q)$, so $\botequiv_{\kappa(Q)-1}$ (and every further relation) is the same as $\botequiv$.

Observe that, for a $k \ge 0$, if $S$ is not contained in a single equivalence class of $\botequiv_k$, then there exists a word $w \in \Sigma^{\le k}$ such that, for some states $p,q \in S$, we have $\delta(p,w)\neq\bot$ and $\delta(q,w)=\bot$, thus $w$ satisfies $1 \le \kappa(\delta(S,w)) < \kappa(S)$.

We consider $\botequiv_{\kappa(Q)-\kappa(S)+1}$.
It has at least $\kappa(Q)-\kappa(S)+2$ equivalence classes.
Therefore, since there are at most $\kappa(Q)-\kappa(S)$ equivalence classes not intersecting $S$, this relation must have at least two classes that intersect $S$, so $S$ is not contained in a single class.
It follows that there is a word $w$ of length at most $\kappa(Q)-\kappa(S)+1$ satisfying the lemma.

Finally, we have $\kappa(Q)-\kappa(S)+1 \le n - |S| + 1$ since there are at most $n-|S|$ equivalence classes in $Q \setminus S$, so $\kappa(Q) \le n-|S|+\kappa(S)$.
\end{proof}

By an iterative application of~\Cref{lem:voiding_lemma}, we can easily compress any subset of states to a subset of a single equivalence class.

\begin{corollary}\label{cor:voiding_lemma}
Let $\mathrsfs{A}=(Q,\Sigma,\delta)$ be a partial DFA, and let $S \subseteq Q$ be a non-empty subset.
There is a word $w$ of length at most $(\kappa(S)-1)(\kappa(Q)-\kappa(S)/2)$ such that $\delta(S,w)$ is non-empty and is contained in one inseparability class.
\end{corollary}
\begin{proof}
In the worst case, we apply at most $\kappa(S)-1$ times \Cref{lem:voiding_lemma} for subsets intersecting $\kappa(S),\kappa(S)-1,\ldots,2$ equivalence classes.
\end{proof}

\subsection{Fixing automaton}

The other possibility of compressing a subset in a partial DFA is the classical pair compression.
This is the only way to compress a subset with all states in one equivalence class, which is always the case in a complete DFA.

Our tool to deal with this way of compression is the \emph{fixing automaton}.
This is a complete DFA obtained from a partial one, defined as follows.

\begin{definition}[Fixing automaton]
For a partial DFA $\mathrsfs{A}(Q,\Sigma,\delta)$, the \emph{fixing automaton} is the complete DFA $\mathrsfs{A}^\mathrm{F}=(Q,\Sigma,\delta^\mathrm{F})$ such that the states are fixed instead of having an undefined transition:
for every $q \in Q$ and $a \in \Sigma$, we have $\delta^\mathrm{F}(q,a) = q$ if $\delta(q,a)=\bot$, and $\delta^\mathrm{F}(q,a)=\delta(q,a)$ otherwise.
\end{definition}

We list some useful properties of the fixing automaton.

\begin{lemma}\label{lem:fixing_same_action}
Let $\mathrsfs{A}=(Q,\Sigma,\delta)$ be a partial DFA, let $S \subseteq Q$, and let $w \in \Sigma^*$.
We have $\delta(S,w) \subseteq \delta^\mathrm{F}(S,w)$.
Moreover, if for every state $q \in S$ we have $\delta(q,w) \neq \bot$, then $\delta(S,w)=\delta^\mathrm{F}(S,w)$.
\end{lemma}

\begin{lemma}\label{lem:fixing_word}
Let $\mathrsfs{A}=(Q,\Sigma,\delta)$ be a partial DFA and let $S \subseteq Q$ be a non-empty subset.
For every word $w \in \Sigma^*$, there exists a word $w' \in \Sigma^*$ of length $|w'| \le |w|$ such that $\emptyset \neq \delta(S,w') \subseteq \delta^\mathrm{F}(S,w)$.
In particular, if $w$ has rank $r$ in $\mathrsfs{A}^\mathrm{F}$, then $w'$ has rank $1 \le r' \le r$ in $\mathrsfs{A}$.
\end{lemma}
\begin{proof}
For a letter $a \in \Sigma$, let $\delta^{-1}(\bot,a) = \{q \in Q \mid \delta(q,a)=\bot\}$, which is the set of states that are mapped to $\bot$ under the action of $a$ in $\mathrsfs{A}$.

We prove the statement by induction on the length $|w|$.
Obviously, it holds for $|w|=0$. Consider $w = ua$ for some $u \in \Sigma^*$ and $a \in \Sigma$, and let $u'$ be the word obtained from the inductive assumption for~$u$.
Recall that all states from $\delta^{-1}(\bot,a)$ are fixed in $\mathrsfs{A}^\mathrm{F}$ under the action of $a$.
Let $T = \delta(S,u')$; thus $\emptyset \neq T \subseteq \delta^\mathrm{F}(S,u)$ by the inductive assumption.
We have two cases.

\textit{(Case~1)} If $T \subseteq \delta^{-1}(\bot,a)$, then we let $w' = u'$.
Hence, we still have $\delta(S,w') = \delta(S,u') = T \subseteq \delta^\mathrm{F}(S,ua)$, because $T$ is fixed under the action of $a$ by $\delta^\mathrm{F}$.

\textit{(Case~2)} If $T \nsubseteq \delta^{-1}(\bot,a)$, then we let $w' = u'a$.
Since there is a state $q \in T \setminus \delta^{-1}(\bot,a)$, we know that $\delta(T,a)$ is non-empty.
We also have $\delta(T,a) \subseteq \delta^\mathrm{F}(T,a)$, since the states in $T \setminus \delta^{-1}(\bot,a)$ are mapped in the same way under the action of $a$ by both $\delta$ and $\delta^\mathrm{F}$.
We get that $\delta(S,u'a) = \delta(T,a) \subseteq \delta^\mathrm{F}(T,a) \subseteq \delta^\mathrm{F}(S,ua)$.
\end{proof}

\begin{corollary}
The minimal non-zero rank of a partial DFA $\mathrsfs{A}$ is at most the minimal rank of $\mathrsfs{A}^\mathrm{F}$.
\end{corollary}

In the case of a partial DFA that is not strongly connected, it can happen that we cannot compress some subset~$S$ even when there exists a word of non-zero rank smaller than $|S|$. 
This is the reason why the shortest words of the minimal non-zero rank can be exponentially long and why deciding if there is a word of a given rank is \PSPACE-complete~\cite{Berlinkov2014OnTwoAlgorithmicProblems}. 
However, in the case of a strongly connected partial DFA, as well as for a (not necessarily strongly connected) complete DFA, every non-mortal word can be extended to a word of the minimal non-zero rank.
This is a fundamental difference that allows constructing compressing words iteratively.
Note that the fixing automaton of a strongly connected partial DFA is also strongly connected.

\begin{lemma}\label{lem:strcon_minimal_rank}
Let $\mathrsfs{A}=(Q,\Sigma,\delta)$ be a strongly connected partial DFA, and let $r$ be the minimal non-zero rank over all words.
For every non-empty subset $S \subseteq Q$, there exists a non-mortal word $w$ such that $|\delta(S,w)| \le r$.
\end{lemma}
\begin{proof}
Let $u$ be a word of rank $r$.
Then there exists a state $q \in Q$ such that $\delta(q,u)$ is defined ($\neq \bot$).
Let $p \in S$ be any state and let $v_{p,q}$ be a word mapping $p$ to $q$ (such a word always exists because of the strong connectivity).
Then $\delta(S,v_{p,q} u)$ has a non-zero rank $\le r$.
\end{proof}

\subsection{Rank conjecture}

The \emph{rank conjecture} (sometimes called \emph{\v{C}ern\'{y}-Pin conjecture}) is a well-known generalization of the \v{C}ern\'{y} conjecture to non-synchronizing DFAs (e.g., \cite{Pin1983OnTwoCombinatorialProblems}).
The rank conjecture is a weaker version of the conjecture originally stated by Pin that was not restricted to the minimal rank and turned out to be false~\cite{Kari2001Counterexample}.
Some further results on the rank conjecture for strongly connected complete DFAs are provided in~\cite{Kari2019}.

\begin{conjecture}[The rank conjecture]
For an $n$-state complete DFA where $r$ is the minimal rank over all words, there exists a word of rank $r$ and of length at most $(n-r)^2$.
\end{conjecture}

For partial DFAs, the rank conjecture is analogous with the exception that $r$ is the minimal non-zero rank.

\begin{theorem}\label{thm:rank_conjecture_transfer}
Let $\mathrsfs{A}=(Q,\Sigma,\delta)$ be a strongly connected partial DFA.
If the rank conjecture holds true for the fixing automaton $\mathrsfs{A}^\mathrm{F}$, then it also holds for $\mathrsfs{A}$.
\end{theorem}
\begin{proof}
Let $r$ be the minimal rank in $\mathrsfs{A}^\mathrm{F}$ over all words.
From the conjecture and by \Cref{lem:fixing_word}, there exists a word $w'$ of length at most $(n-r)^2$ and such that $\emptyset \neq \delta(Q,w') \subseteq \delta^\mathrm{F}(Q,w)$.

Let $r' \leq r$ be the minimal non-zero rank in $\mathrsfs{A}$ over all words.
For every $s=r,r-1,\ldots,r'+1$, we inductively construct a word of non-zero rank less than $s$, of length at most $(n-(s-1))^2$, and such that~$w'$ is its prefix.
Let $w'v$ be a word of non-zero rank at most $s$ in $\mathrsfs{A}$ and of length at most~$(n-s)^2$, and let $S=\delta(Q,w'v)$.
Suppose that $\kappa(S)=1$.
Since $s$ is not the minimal rank of $\mathrsfs{A}$, by \Cref{lem:strcon_minimal_rank}, $S$ must be compressible.
Since its states are inseparable, there must be two distinct states $p,q \in S$ and a word $u$ such that $\delta(q,u)=\delta(p,u) \neq \bot$.
But (from \Cref{lem:fixing_same_action}) $\{p,q\} \subseteq \delta(Q,w'v) \subseteq \delta^\mathrm{F}(Q,w'v) \subseteq \delta^\mathrm{F}(Q,wv)$, thus $\delta^\mathrm{F}(Q,wv)$ is compressible in $\mathrsfs{A}^\mathrm{F}$, which contradicts the fact that $w$ has the minimal rank in $\mathrsfs{A}^\mathrm{F}$.
Hence $\kappa(S) \ge 2$, and by \Cref{lem:voiding_lemma}, $\delta(Q,w'v)$ can be compressed with a word $u$ of length at most $n-s+1$.
We have $|w'vu| \le (n-s)^2 + n-s+1 \le (n-(s-1))^2$, which proves the induction step.
\end{proof}

The theorem implies that, in the strongly connected case, the rank conjecture is true for complete DFAs if and only if it is true for partial DFAs.
It also immediately implies the following result for the class of Eulerian DFAs.
A partial DFA is \emph{Eulerian} if it is strongly connected and the numbers of outgoing and incoming transitions are the same at every state, i.e., for every $q \in Q$, we have $|\{a \in \Sigma \mid \delta(q,a) \ne \bot\}| = |\{(p,a) \in Q \times \Sigma \mid \delta(p,a)=q\}|$.
The following corollary follows from the facts that the rank conjecture holds for complete Eulerian DFAs \cite{Kari2019} and that the fixing automaton of a partial Eulerian DFA is also Eulerian.

\begin{corollary}\label{cor:eulerian}
The rank conjecture is true for partial Eulerian DFAs.
\end{corollary}

\subsection{Collecting automaton}

The fixing automaton allows us to analyze the behavior of words in a partial DFA by studying a complete DFA. However, its main disadvantage is that the fixing automaton of a synchronizing partial DFA is not necessarily synchronizing.
Therefore, we will need one more tool, called \emph{collecting automaton}.
It is an extension of the fixing automaton by an additional letter that allows quick synchronization into one inseparability class, while it does not affect the length of a shortest synchronizing word for any particular inseparability class.

By $\mathrsfs{A}/_{\botequiv} = (Q_{\mathrsfs{A}/_{\botequiv}},\Sigma,\delta_{\mathrsfs{A}/_{\botequiv}})$, we denote the quotient DFA by the inseparability relation (clearly, this relation is a congruence).
$\mathrsfs{A}/_{\botequiv}$ is also a partial DFA, and if $\mathrsfs{A}$ is strongly connected, then so is $\mathrsfs{A}/_{\botequiv}$.
By $[p] \in Q_{\mathrsfs{A}/_{\botequiv}}$, we denote the equivalence class of a state $p \in Q$ of the original DFA $\mathrsfs{A}$.

A \emph{collecting tree} of $\mathrsfs{A}$ is a tree $T$ with the set of vertices $Q_{\mathrsfs{A}/_{\botequiv}}$ and directed edges labeled by letters from $\Sigma$ in the following way:
\begin{enumerate}
\item[(a)] edges and their labels correspond to transitions in $\mathrsfs{A}/_{\botequiv}$:\\each edge $([p],a,[p'])$ is such that ${\delta_{\mathrsfs{A}/_{\botequiv}}([p],a) = [p']}$;
\item[(b)] there is a root $[r]$ such that the tree is directed toward it.
\end{enumerate}
See \Cref{fig:collecting} for an example.
Equivalently, a collecting tree can be seen as a specific partial DFA which is a sub-automaton of $\mathrsfs{A}/_{\botequiv}$ whose underlying digraph is a tree directed toward one state.
A DFA can have many collecting trees, even for the same $[r]$, and every strongly connected DFA has a collecting tree for every class $[r]$.

\begin{figure}[htb]\centering
\raisebox{-0.5\height}{\includegraphics[scale=1.2]{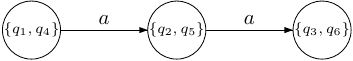}}\hspace{2cm}
\raisebox{-0.5\height}{\includegraphics{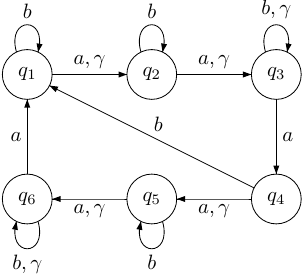}}
\caption{\normalsize Left: a collecting tree with root $[q3]=\{q_3,q_6\}$; right: the corresponding collecting automaton of the example from \Cref{fig:examples} (left).}
\label{fig:collecting}
\end{figure}

\begin{definition}[Collecting automaton]
Let $\mathrsfs{A}=(Q,\Sigma,\delta)$ be a strongly connected partial DFA, and let~$T$ be one of its collecting trees with a root $[r]$.
The \emph{collecting automaton} $\mathrsfs{A}^{\mathrm{C}(T)}=(Q,\Sigma \cup \{\gamma\}, \delta^{\mathrm{C}(T)})$ is defined as follows:
\begin{itemize}
\item The transition function $\delta^{\mathrm{C}(T)}$ on $\Sigma$ is defined as in the fixing automaton $\mathrsfs{A}^\mathrm{F}$.
\item $\gamma \notin \Sigma$ is a fresh letter. Its action is defined according to the edges in $T$:
Let $q_1 \in Q \setminus [r]$ be a state.
Since $T$ is a tree directed toward $[r]$, there is exactly one edge outgoing from $[q_1]$, say $([q_1],a,[q_2]) \in T$ for some $[q_2] \in Q_{\mathrsfs{A}/_{\botequiv}}$ and $a \in \Sigma$.
We set $\delta^{\mathrm{C}(T)}(q_1,\gamma)=\delta(q_1,a)$.
Finally, the transition of $\gamma$ on each state in $[r]$ is a self-loop.
\end{itemize}
\end{definition}

A collecting automaton of a strongly connected partial DFA is always strongly connected, as it contains all transitions of the original partial DFA.
It is also always complete.
We prove several properties connecting partial DFAs and their collecting automata.
They are preliminary steps toward relating the \v{C}ern\'{y} conjecture for strongly connected partial and complete DFAs.

\begin{lemma}\label{lem:synchro_class_by_collecting}
Let $\mathrsfs{A}=(Q,\Sigma,\delta)$ be a strongly connected partial DFA, and let $T$ be one of its collecting trees with a root $[r]$.
If there is a word over $\Sigma \cup \{\gamma\}$ that synchronizes $[r]$ in $\mathrsfs{A}^{\mathrm{C}(T)}=(Q,\Sigma \cup \{\gamma\},\delta^{\mathrm{C}(T)})$, then there is also such a word over $\Sigma$ of at most the same length.
\end{lemma}
\begin{proof}
Let $w' \in (\Sigma \cup \{\gamma\})^*$ be a word that synchronizes $[r]$ in $\mathrsfs{A}^{\mathrm{C}(T)}$.
We construct $w \in \Sigma^*$ from~$w'$ such that~$w$ also synchronizes~$[r]$ and $|w| \le |w'|$.
This is done by replacing every occurrence of $\gamma$ in~$w'$ with a suitable substituting letter from $\Sigma$ or removing it.

Consider an occurrence of $\gamma$ in $w'$, so let $w' = u \gamma v$, for some $u,v \in (\Sigma \cup \{\gamma\})^*$.
Let $[p]$ be the class such that $\delta^{\mathrm{C}(T)}([r],u) \subseteq [p]$.
If $[p] \neq [r]$, then let $a \in \Sigma$ be the letter labeling the outgoing edge $([p],a,[p'])$ in the collecting tree.
The action of $a$ on $[p]$ is thus the same as the action of $\gamma$.
Hence, $\delta^{\mathrm{C}(T)}([r],u \gamma v) = \delta^{\mathrm{C}(T)}([r],u a v)$.
If $[p] = [r]$, then $\gamma$ acts as identity on $[p]$, thus $\delta^{\mathrm{C}(T)}([r],u \gamma v) = \delta^{\mathrm{C}(T)}([r],u v)$.

By considering the occurrences of $\gamma$ in any order and performing the described changes, we obtain a word $w$ such that $\delta^{\mathrm{C}(T)}([r],w') = \delta^{\mathrm{C}(T)}([r],w)$. Thus $w$ synchronizes $[r]$, since so does~$w'$.
\end{proof}

\begin{lemma}\label{lem:collecting_automaton_synchro}
Let $\mathrsfs{A}=(Q,\Sigma,\delta)$ be a strongly connected partial DFA and let $T$ be one of its collecting trees.
Then $\mathrsfs{A}$ is synchronizing if and only if the collecting automaton $\mathrsfs{A}^{\mathrm{C}(T)}=(Q,\Sigma \cup \{\gamma\},\delta^{\mathrm{C}(T)})$ is synchronizing.
\end{lemma}
\begin{proof}
Let $[r]$ be the root class of $T$.
Let $w \in \Sigma^*$ be a reset word for $\mathrsfs{A}$.
Then there is some class $[p]$ such that $w$ synchronizes it.
By \Cref{lem:fixing_same_action}, $\delta([p],w)=\delta^{\mathrm{F}}([p],w)$ is a singleton.
Since $\mathrsfs{A}$ is strongly connected, there is a word $u_{r \to p}$ whose action maps all states from $[r]$ into $[p]$.
Then $\gamma^{n-1} u_{r \to p} w$ is a synchronizing word for $\mathrsfs{A}^{\mathrm{C}(T)}$.

Conversely, if $\mathrsfs{A}^{\mathrm{C}(T)}$ is synchronizing, then there exists a word synchronizing $[r]$.
By \Cref{lem:synchro_class_by_collecting}, there is also a word $w$ synchronizing $[r]$ in $\mathrsfs{A}^{\mathrm{C}(T)}$.
By \Cref{cor:voiding_lemma}, we can find a word $v$ whose action maps $Q$ into a single equivalence class $[p]$.
By strong connectivity, we can find a word $u_{p \to r}$ whose action maps all states from $[p]$ into $[r]$.
Thus, $v u_{p \to r} w$ is a reset word for $\mathrsfs{A}$.
\end{proof}

\Cref{lem:collecting_automaton_synchro} implies in particular that the choice of $T$ does not matter: $\mathrsfs{A}=(Q,\Sigma,\delta)$ is synchronizing if and only if all $\mathrsfs{A}^{\mathrm{C}(T)}$ are synchronizing.

\subsection{Algorithmic aspects}\label{sec:alg_issues}

Checking if a strongly connected partial DFA is synchronizing and finding a word of minimum non-zero rank can be done similarly as for a complete DFA, by a suitable generalization of the well-known Eppstein algorithm \cite[Algorithm~1]{Ep1990}.
The same algorithm for checking synchronizability, under different terminology, was described in the context of $\varepsilon$-machines~\cite{TrCr11Exact}.

\begin{proposition}\label{pro:synchronization_algorithm}
Checking if a given strongly connected partial DFA with $n$ states over an alphabet~$\Sigma$ is synchronizing can be done in $\O(|\Sigma|\cdot n^2)$ time and $\O(n^2+|\Sigma|\cdot n)$ space.
Finding a word of minimum non-zero rank can be done in $\O(|\Sigma|\cdot n^3)$ time and $\O(n^2+|\Sigma|\cdot n)$ space (not counting the word's space itself).
\end{proposition}
\begin{proof}
We modify the well-known Eppstein algorithm \cite[Algorithm~1]{Ep1990} by generalizing it to partial and not necessarily synchronizing automata.
We mark all pairs of states $\{p,q\}$ such that $|\delta(\{p,q\},a)|=1$ for some $a \in \Sigma$; note that this includes two possibilities of compression, which is the only difference from the standard algorithm.
Then we propagate the compressibility backwards using a breadth-first search:
For a compressible pair $\{p,q\}$, we mark all pairs $\{p',q'\}$ such that there exists a letter $a \in \Sigma$ with $\delta(p',a)=p$ and $\delta(q',a)=q$ (or dually $\delta(p',a)=q$ and $\delta(q’,a)=p$, since the pairs are unordered).
The partial DFA is synchronizing if and only if all pairs are compressible.
The time complexity of this is $\O(|\Sigma|\cdot n^2)$.

To find a word of minimum non-zero rank, we repetitively apply a shortest word that compresses a pair of states in the current subset, starting from $Q$.
This part works in the same way as in the Eppstein algorithm, including the trick to reduce the space complexity to quadratic.
Each time we apply a word that compresses a pair of states, and it is guaranteed that at least one of these states does not go to $\bot$.
The algorithm stops when the current subset is no longer compressible, and by \Cref{lem:strcon_minimal_rank}, its size is equal to the minimal non-zero rank.
\end{proof}

Furthermore, there exists an even more efficient reduction from the problem of checking if a strongly connected partial DFA is synchronizing to the problem of checking if a complete DFA is synchronizing. Hence, the existence of a more efficient algorithm for the latter would imply the existence of a more efficient algorithm for the former.

\begin{theorem}\label{thm:synchronizability_reduction}
Given a strongly connected partial DFA with $n$ states over an alphabet $\Sigma$, in $\O(|\Sigma|\cdot n\log n)$ time, we can construct a complete DFA that is synchronizing if and only if the given partial DFA is synchronizing.
\end{theorem}
\begin{proof}
We can compute all inseparability classes in $\O(|\Sigma|\cdot n \log n)$ time.
This is done by the Hopcroft minimization algorithm \cite{Hopcroft1971}, if we interpret the partial DFA as a language-accepting DFA with an arbitrary initial state and a new sink state $\bot$ that is its only final state.

Having computed the classes, we can construct a collecting automaton for an arbitrary collecting tree.
Note that it can be done in $\O(|\Sigma|\cdot n)$ time by a breadth-first search from a class $[r]$.
The desired property follows from~\Cref{lem:collecting_automaton_synchro}.
\end{proof}

\subsection{Induced automaton}\label{sec:induced}

We develop an algebraic technique applied to partial DFAs.
It will allow us to derive upper bounds on reset thresholds, in particular, in the cases when there exists a short word of a small rank, which is the case, for example, for the literal automaton of a prefix code. We also use the results of this subsection in the next subsection to translate upper bounds on the reset threshold from complete DFAs to partial DFAs.
We base on the results from~\cite{BS2016AlgebraicSynchronizationCriterion} for complete DFAs and generalize them to be applied to partial DFAs.
The existing linear algebraic proofs for complete DFAs do not work for partial ones, because the transition matrices may not have a constant sum of the entries in each row.
Furthermore, our generalization simplifies the previous proof, which, to show that an induced automaton is synchronizing, uses the stationary distribution of a Markov chain defined by the automaton and the extension method (applying words that yield a larger preimage of a subset).

We need to introduce a few definitions from linear algebra for automata (see, e.g.,~\cite{BS2016AlgebraicSynchronizationCriterion,Kari2003Eulerian,Pin1972Utilisation,Szykula2018ImprovingTheUpperBound}).
Let $\mathrsfs{A} = (Q,\Sigma,\delta)$ be a partial DFA.
Without loss of generality we assume that $Q = \{1,\ldots,n\}$.
By~$\mathbb{Q}^n$, we denote the rational\footnote{The method works over the real field $\mathbb{R}$ too, yet $\mathbb{Q}$ is sufficient.} $n$-dimensional linear space of row vectors.
For a vector $v \in \mathbb{Q}^n$ and an $i \in Q$, we denote the vector's value at the $i$-th position by $v(i)$.
Similarly, for a matrix $M$, we denote its value in the $i$-th row and the $j$-th column by $M(i,j)$.
A vector $g$ is \emph{non-negative} if $g(i)\ge 0$ for all $i$, and it is \emph{non-zero} if $g(i) \neq 0$ for some $i$.
For a word $w \in \Sigma^*$, by $\mathrm{M}(w)$ we denote the $n \times n$ matrix of the transformation of $w$ in $\delta$:
$\mathrm{M}(w)(p,q) = 1$ if $\delta(p,w)=q$, and $\mathrm{M}(w)(p,q)=0$ otherwise.
Note that if $\delta(p,w)=\bot$, then we have $\mathrm{M}(w)(p,q)=0$ for all $q \in Q$.
The usual scalar product of two vectors $u,v$ is denoted by $u \scalar v$.
The linear subspace spanned by a set of vectors $V$ is denoted by $\lspan(V)$.

Given a transition function $\delta$ (which defines matrices $\mathrm{M}(w)$), call a set of words $W \subseteq \Sigma^*$ \emph{complete for a subspace} $V \subseteq \mathbb{Q}^n$ \emph{with respect to a vector} $g \in V$, if $V \subseteq \lspan(\{g \mathrm{M}(w) \mid w \in W\})$.
A set of words $W \subseteq \Sigma^*$ is \emph{complete for a subspace} $V \subseteq \mathbb{Q}^n$ if for every non-negative non-zero vector $g \in V$, $W$ is complete for $V$ with respect to $g$.
Let $\chi(p)$ denote the characteristic (unitary) vector of $\{p\}$.
For a subset $S \subseteq Q$, we define $\mathbb{V}(S) = \lspan(\{\chi(p) \mid p \in S\}) \subseteq \mathbb{Q}^n$.

For example, consider the DFA from~\Cref{fig:examples}~(left).
Let $V = \mathbb{V}(\{q_1,q_2,q_5\})$ and $W = \{ab,aab\}a^{\le 5}$.
Let $g \in V$ be a non-negative non-zero vector, and let $i$ be such that $g(i) \neq 0$.
If $i=1$ then let $u=ab$, and otherwise let $u=aab$; then $g \mathrm{M}(u)$ has exactly one non-zero entry.
Then, for each $j \in \{1,2,5\}$, the vector $g \mathrm{M}(u a^{j'})$ for some $j'$ has the unique non-zero entry at $q_j$.
These vectors generate $V$, thus $W$ is complete for $V$ with respect to $g$.

The induced automaton of a partial DFA is a partial DFA acting on a subset of states $R \subseteq Q$.
It is built from two sets of words.
Let $W_1$ be a set of words such that $R = \bigcup_{w \in W_1} \delta(Q,w)$.
Thus, for each state in~$R$, there is some state mapped to it by a word from $W_1$.
Intuitively, the second set $W_2$ is any non-empty set of words that enriches the actions of words from $W_1$.
The induced automaton is $\mathrsfs{A}$ restricted to $R$ with alphabet $W_2 W_1$.
Note that its transition function is well defined, which is ensured by the fact that every word of the form $w_2 w_1 \in W_2 W_1$ has the action mapping every state $q \in Q$ into $R$ or to~$\bot$.

\begin{definition}[Induced automaton]\label{def:induced_automaton}
Let $W_1,W_2 \subseteq \Sigma^*$ be non-empty and $R = \{\delta(q,w) \mid q \in Q,\ w \in W_1, \delta(q,w) \neq \bot \}$.
If $R$ is non-empty, we define the \emph{induced automaton}
\[ \mathrsfs{A}^{\mathrm{I}(W_1,W_2)}=(R, W_2 W_1, \delta_{\mathrsfs{A}^{\mathrm{I}(W_1,W_2)}}), \]
where the transition function is defined in compliance with the actions of words in $\mathrsfs{A}$, i.e.,\\ ${\delta_{\mathrsfs{A}^{\mathrm{I}(W_1,W_2)}}(q,w) = \delta(q,w)}$ for all $q \in R$ and $w \in W_2 W_1$.
\end{definition}

We can analyze an induced automaton as a separate one, and synchronize the whole DFA using it, which is particularly profitable when $R$ is small.
Following our previous example, for  \Cref{fig:examples} (left) with $W_1 = \{b\}$ and $W_2 = \{ab,aab\}a^{\le 5}$ we obtain the induced automaton on $R=\{q_1,q_2,q_5\}$.
Furthermore, it is synchronizing already by a letter from $W_2 W_1$ (e.g., $abb$), and each reset word corresponds to a reset word of the original $\mathrsfs{A}$.

The following lemma states that the completeness of a set of words together with the synchronizability and strong connectivity of the whole DFA transfer to the induced automaton.
It generalizes~\cite[Theorem~2]{BS2016AlgebraicSynchronizationCriterion} to partial DFAs, and the proof uses a recursion instead of an augmenting argument.

\begin{lemma}\label{lem:induced_completeness}
Let $\mathrsfs{A}=(Q,\Sigma,\delta)$ be a strongly connected synchronizing partial DFA and let $W_1$ and~$W_2$ be two non-empty sets of words over $\Sigma$.
Let $\mathrsfs{A}^{\mathrm{I}(W_1,W_2)} = (R, W_2 W_1, \delta_{\mathrsfs{A}^{\mathrm{I}(W_1,W_2)}})$ be the induced automaton of $\mathrsfs{A}$.
If $W_2$ is complete for $\mathbb{V}(Q)=\mathbb{Q}^n$, then $W_2 W_1$ is complete for $\mathbb{V}(R)$, and $\mathrsfs{A}^{\mathrm{I}(W_1,W_2)}$ is synchronizing and strongly connected.
\end{lemma}
\begin{proof}
First, let us prove the completeness of $W_2 W_1$ for $\mathbb{V}(R)$ with respect to an arbitrary non-negative non-zero vector $\alpha \in \mathbb{V}(R)$ with rational entries.
By the definition of $R$, for each state $r \in R$, there is a state $q_1$ and a word $u_1 \in W_1$ such that $\delta(q_1, u_1) = r$.
Since $\mathrsfs{A}=(Q,\Sigma,\delta)$ is synchronizing and strongly connected, there is a synchronizing word $u_2$ such that $\delta(Q, u_2) = \{q_1\}$.
Since the image is non-empty, there is a state $q_2 \in Q$ such that $\delta(q_2, u_2) = q_1$.
Let $q_3 \in R$ be a state such that $\chi(q_3) \scalar \alpha > 0$ (it exists since $\alpha \in \mathbb{V}(R)$ is non-zero).
Again by the strong connectivity of $\mathrsfs{A}$, there is a word $u_3$ whose action maps $q_3$ to $q_2$.
Hence, $u_3 u_2 u_1$ is a synchronizing word with the action mapping $q_3$ to~$r$.
Thus $\alpha\cdot\mathrm{M}(u_3 u_2 u_1) \in \lspan(\{\chi(r)\})$, and since $\alpha$ and $\mathrm{M}(u_3 u_2 u_1)$ are also non-negative, we have $(\alpha\cdot \mathrm{M}(u_3 u_2 u_1))(r) \neq 0$.
Moreover, $(\alpha\cdot\mathrm{M}(u_3 u_2 u_1))(r') = 0$ for all $r' \in R \setminus \{r\}$.
Since for every state $r \in R$ we can find such a word $u_3u_2u_1$, and $u_1 \in W_1$, we have $\lspan(\{\alpha\cdot \mathrm{M}(w) \mid w \in \Sigma^* W_1\}) = \mathbb{V}(R)$.
As $W_2$ is complete for $\mathbb{V}(Q)$, for each $w \in \Sigma^*$, we have 
$\alpha\cdot\mathrm{M}(w) \in \lspan(\{\alpha\cdot\mathrm{M}(w_2) \mid w_2 \in W_2\})$.
This implies that $\lspan(\{\alpha\cdot \mathrm{M}(w) \mid w \in W_2 W_1\}) = \lspan(\{\alpha\cdot \mathrm{M}(w) \mid w \in \Sigma^* W_1\}) = \mathbb{V}(R)$.
Hence, we get that $W_2W_1$ is complete for $\mathbb{V}(R)$ with respect to an arbitrary non-negative non-zero vector $\alpha \in \mathbb{V}(R)$ with rational entries.

It remains to prove that $\mathrsfs{B} = \mathrsfs{A}^{\mathrm{I}(W_1,W_2)}$ is synchronizing.
If $|R| = 1$, we are done, so consider $|R| \ge 2$.
Suppose that $R$ is not compressible in $\mathrsfs{B}$.
Then, for each word $w \in W_2 W_1$, we have $\delta(R, w) = R$ or $\delta(R,w) = \emptyset$.
This contradicts with $W_2 W_1$ being complete for $\mathbb{V}(R)$ with respect to $\chi(R)$ since in this case $\lspan(\{\chi(R)\cdot\mathrm{M}(w) \mid w \in W_2 W_1\}) \subseteq \lspan(\{\chi(R)\}) \subsetneq \mathbb{V}(R)$.
Hence $R$ is compressible in $\mathrsfs{B}$ and the statement follows by induction.
Note that the conditions of the lemma are met for the induced automaton $\mathrsfs{A}^{\mathrm{I}(W_1 \{w\},W_2)} = (\delta(R,w),W_2 W_1 \{w\},\delta_{\mathrsfs{A}^{\mathrm{I}(W_1 \{w\},W_2)}})$, where $w \in W_2 W_1$ is a word that compresses $R$.
\end{proof}

The following corollary directly follows from \Cref{lem:induced_completeness}, since $\Sigma^{\leq n-1}$ is always complete for $\mathbb{V}(Q)$ in the case of a strongly connected synchronizing partial DFA.

\begin{corollary}\label{cor:rt_by_induced}
Let $\mathrsfs{A}=(Q,\Sigma,\delta)$ be a strongly connected synchronizing partial DFA with $n$ states, and let $w \in \Sigma^*$ be a word such that $R = \delta(Q,w) \neq \emptyset$.
Let $W_1 = \{w\}$, $W_2 = \Sigma^{\le n-1}$, and $\mathrsfs{A}^{\mathrm{I}(W_1,W_2)} = (R,\Sigma^{\leq n-1}\{w\},\delta_{\mathrsfs{A}^{\mathrm{I}(W_1,W_2)}})$ be the induced automaton. Then
\[\rt(\mathrsfs{A}) \leq |w| + (|w| + n-1)\cdot\rt(\mathrsfs{A}^{\mathrm{I}(W_1,W_2)}).\]
\end{corollary}
\begin{proof}
Let $\alpha$ be a non-negative non-zero vector, and let $q \in Q$ be such that $\chi(q) \scalar \alpha > 0$.
Since the DFA is strongly connected and synchronizing, for every $p \in Q$, there is a word $u_p$ such that $\delta(Q,u_p) = \{p\}$ and furthermore $\delta(q,u_p) = p$.
It follows that $\lspan(\{\alpha\cdot\mathrm{M}(w) \mid w \in \Sigma^*\}) = \mathbb{V}(Q)$.
By the usual ascending chain argument (e.g., \cite{Kari2003Eulerian,Pin1972Utilisation,Szykula2018ImprovingTheUpperBound}), $\lspan(\{\alpha\cdot\mathrm{M}(w) \mid w \in \Sigma^{\leq n-1}\}) = \lspan(\{\alpha\cdot\mathrm{M}(w) \mid w \in \Sigma^*\}) = \mathbb{V}(Q)$.
Thus, we can apply \Cref{lem:induced_completeness} for the induced automaton.
\end{proof}

The corollary is useful for deriving upper bounds for DFAs with a word of small rank.
Having such a word $w$, we can further synchronize $R$ through the induced automaton instead of trying to do this directly.
Although every letter of $\mathrsfs{A}^{\mathrm{I}(W_1,W_2)}$ corresponds to a word of length $|w|+n-1$ in the original $\mathrsfs{A}$, if $R$ is small enough, this yields a better upper bound.
We show its application in the next subsection.

\subsection{The \texorpdfstring{\v{C}ern\'{y}}{Cerny} conjecture}

The famous \v{C}ern\'{y} conjecture is a special case of the rank conjecture for complete DFAs for $r=1$.
Let $\mathfrak{C}(n)$ be the maximum length of the shortest reset words of all $n$-state synchronizing complete DFAs.
It is well known that $\mathfrak{C}(n) \ge (n-1)^2$ \cite{Cerny1964}.
The \v{C}ern\'{y} conjecture states that $\mathfrak{C}(n) = (n-1)^2$, but the best proved upper bound is cubic \cite{Shitov2019,Szykula2018ImprovingTheUpperBound}.

Let $\mathfrak{C}_\mathrm{P}(n)$ be the maximal length of the shortest reset words of all $n$-state synchronizing strongly connected partial DFAs.
We show that an improvement of the upper bound on $\mathfrak{C}(n)$ implies a similar (but slightly weaker) improvement on $\mathfrak{C}_\mathrm{P}(n)$.
In particular, the \v{C}ern\'{y} conjecture implies that $\mathfrak{C}_\mathrm{P}(n) \le n^2 + \O(n)$.

To prove the following theorem, we combine several techniques, in particular, the inseparability equivalence, the collecting automaton, and an algebraic upper bound on the reset threshold of a complete DFA with a short word of small rank~\cite{BS2016AlgebraicSynchronizationCriterion}.

\begin{theorem}\label{thm:cerny_transfer}
Suppose that for all $n$, $\mathfrak{C}(n) \le n^k$, for some fixed $2 \le k \le 3$.
Then:
\[ \mathfrak{C}_\mathrm{P}(n) \le \mathfrak{C}(n) + \O(n^{2-2/k}) \le n^k + o(n^k) .\]
\end{theorem}
\begin{proof}
Let $\mathrsfs{A} = (Q,\Sigma,\delta)$ be a synchronizing partial DFA with $n$ states.
Let $T$ be a collecting tree of~$\mathrsfs{A}$ with a root class $[r]$ containing the smallest number of states.
We consider the collecting automaton $\mathrsfs{A}^{\mathrm{C}(T)} = (Q,\Sigma \cup \{\gamma\},\delta^{\mathrm{C}(T)})$.
By \Cref{lem:collecting_automaton_synchro}, $\mathrsfs{A}^{\mathrm{C}(T)}$ is synchronizing.
We have two cases, depending on the number $\kappa(Q)$ of inseparability classes~of~$\mathrsfs{A}$.

First, suppose that $\kappa(Q) \le 2 n^{1-1/k}$.
Then, by \Cref{cor:voiding_lemma} (for $S=Q$), there is a word $v$ of length at most
\[ (\kappa(Q)-1)\kappa(Q)/2 < \kappa(Q)^2/2 \le 2 n^{2-2/k} \]
such that $\delta(Q,v)$ is non-empty and is contained in one equivalence class, say $[p]$.
Since $\mathrsfs{A}$ is strongly connected, there is a word $u_{p \to r}$ of length at most $n-1$ whose action maps $[p]$ into $[r]$.
Let $w'$ be a reset word for $\mathrsfs{A}^{\mathrm{C}(T)}$ of length at most $\mathfrak{C}(n)$.
In particular, $w'$ synchronizes $[r]$, so by \Cref{lem:synchro_class_by_collecting}, we get a word $w$ of length at most $\mathfrak{C}(n)$ that synchronizes $[r]$ in $\mathrsfs{A}$.
Then, $v u_{p \to r} w$ is a reset word for $\mathrsfs{A}$ of length at most $2 n^{2-2/k} + n-1 + \mathfrak{C}(n) = \mathfrak{C}(n) + \O(n^{2-2/k})$.

In the second case, we have $\kappa(Q) > 2 n^{1-1/k}$.
Then the size of $[r]$, which has been chosen to have the smallest size, is at most $\frac{1}{2}n^{1/k}$.
Thus, $\gamma^{n-1}$ is a word of rank at most $\frac{1}{2}n^{1/k}$.
Then we apply~\Cref{cor:rt_by_induced} (cf.\
\cite[Theorem~2]{BS2016AlgebraicSynchronizationCriterion}) for $\mathrsfs{A}^{\mathrm{C}(T)}$ with this word, obtaining that the reset threshold of $\mathrsfs{A}^{\mathrm{C}(T)}$ is upper bounded by $(n-1) + 2(n-1)\cdot\mathfrak{C}(\lfloor \frac{1}{2} n^{1/k} \rfloor)$.
Using the assumed inequality $\mathfrak{C}(\lfloor 2 n^{1/k} \rfloor) \le (\frac{1}{2} n^{1/k})^k$, we get that there is a reset word $w'$ for $\mathrsfs{A}^{\mathrm{C}(T)}$ of length at most
\[ (n-1)+2(n-1)\cdot(\frac{1}{2}n^{1/k})^k = (n-1)+2(n-1)\cdot \frac{1}{2^k}n.\]
Now, we return to $\mathrsfs{A}$.
By \Cref{cor:voiding_lemma} (for $S=Q$), we get a word $v$ of length at most $(n-1)n/2$ such that $\delta(Q,u)$ is non-empty and is contained in one equivalence class $[p]$.
As in the first case, there is a word $u_{p \to r}$ of length at most $n-1$ whose action maps $[p]$ into $[r]$.
By \Cref{lem:synchro_class_by_collecting}, from~$w'$ we obtain a word $w$ that synchronizes $[r]$ and its length is also at most $(n-1)+2(n-1)\cdot \frac{1}{2^k}n$.
Finally, $v u_{p \to r} w$ is a reset word for $\mathrsfs{A}$ of length at most $(n-1)n/2 + (n-1) + (n-1)+2(n-1)\cdot \frac{1}{2^k}n$.
Since $k \ge 2$, we have:
\[ (n-1)n/2 + (n-1) + (n-1)+2(n-1)\cdot \frac{1}{2^k}n \le n^2/2 + 2(n-1) + 2 n^2 / 4 = n^2 + 2(n-1) \le \mathfrak{C}(n) + \O(n) .\]

From both cases and since $2 \le k \le 3$, we conclude that $\rt(\mathrsfs{A}) \le \mathfrak{C}(n) + \O(n^{2-2/k})$.
\end{proof}

Note that, in particular, if $\mathfrak{C}(n) = (n-1)^2$ then the extra component is linear.
For $\mathfrak{C}(n) \le n^3$, it is $+ \O(n^{4/3})$.

From \Cref{thm:cerny_transfer}, it follows that all upper bounds on the reset threshold of a complete DFA transfer to upper bounds for partial DFAs, up to a subquadratic component.
Thus $\mathfrak{C}_\mathrm{P}(n) \le 0.1654 n^3 + \O(n^2)$ \cite{Shitov2019}.
The extra component is likely not needed, but it is difficult to completely get rid of it in general, as for that we could not lengthen by any means the reset word assumed for a complete DFA.
However, it is easy to omit it when reproving particular bounds for complete DFAs, both combinatorial~\cite{Pin1983OnTwoCombinatorialProblems} and based on avoiding words~\cite{Shitov2019,Szykula2018ImprovingTheUpperBound}.
We conjecture that $\mathfrak{C}_\mathrm{P}(n) = \mathfrak{C}(n)$ for all $n$.

\subsection{The literal automaton of a finite prefix code}\label{sec:literal_automata}

We now use the obtained results about induced automata to get better bounds for partial literal automata of finite prefix codes (defined in~\Cref{subs:stars}).
To do so, we first need to prove that such automata admit short enough words of small rank.
It is known that every complete literal automaton over an alphabet~$\Sigma$ has a word of rank and length at most $\lceil\log_{|\Sigma|} n\rceil$ (\cite[Lemma~16]{BS2016AlgebraicSynchronizationCriterion}, cf.\ \cite[Lemma~14]{BiskupPlandowski2009HuffmanCodes}).
However, this is no longer true for non-mortal words in partial literal automata \cite{RS2018FindingShortSyncWords} and no similar statement was known for any wider class than complete literal automata.
We prove that there exist $\O(\log n)$-rank non-mortal words of length $\O(n)$ in such automata, excluding the case of a code with only one word.
Then we use this result to provide an $\O(n \log^3 n)$ upper bound on the reset threshold of $n$-state synchronizing partial literal automata, asymptotically matching the known upper bound for complete literal automata \cite{BS2016AlgebraicSynchronizationCriterion}.

We start with a special case of one-word codes.
A non-empty word $w$ is called \emph{primitive} if it is not a power of a shorter word, i.e., $w \neq u^k$ for every word $u$ and $k \ge 2$. The upper bound on $\rt(\mathrsfs{A}_X)$ follows from a result of Weinbaum (\cite{Harju2006,Weinbaum1990}).

\begin{proposition}\label{pro:code_size_1}
Let $X = \{x\}$ be a one-word prefix code, and suppose that $x = y^k$, where $y$ is a non-empty primitive word and $k \ge 1$.
Then $\mathrsfs{A}_X$ has rank $k$.
If $\mathrsfs{A}_X$ is synchronizing, then $\rt(\mathrsfs{A}_X) \le \frac{|x|}{2}$, and this bound is tight.
\end{proposition}
\begin{proof}
Observe that the literal automaton of a one-word code forms as a single cycle labeled by the letters of this word.
No two states can be mapped to a single state by the action of any letter, which excludes pair compression.

Observe that $\mathrsfs{A}_X/_{\botequiv}$ is the literal automaton $\mathrsfs{A}_Y$ of $Y = \{y\}$.
In $\mathrsfs{A}_X$, there are $n/k$ inseparability classes, each of size $k$.
Hence, the rank of $\mathrsfs{A}_X$ equals $k$ times the rank of $\mathrsfs{A}_Y$.
Since $y$ is primitive, its action is defined only for the root state of $\mathrsfs{A}_Y$ (and maps the root to itself).
Thus $\mathrsfs{A}_Y$ is synchronizing, so the rank of $\mathrsfs{A}_X$ equals $k$.

Now, we will bound the length of the shortest reset words.
A word $t$ is called a \emph{conjugate} of $t'$ if $t = uv$ and $t' = vu$ for some words $u,v$.
By a result of Weinbaum, stated in terms of automata, every primitive word $x$ has a conjugate $x' = uv$ such that both $u$ and $v$ have the action defined exactly for one state of~$\mathrsfs{A}_X$~\cite{Harju2006}.
Thus, both $u$ and $v$ are reset for $\mathrsfs{A}_X$.
The shorter of them has length at most $\frac{|x|}{2}$.

To see that the bound can be met, consider the code $Z = \{a^k b a^{k + 1} b\}$.
The shortest reset word for $\mathrsfs{A}_Z$ is $a^{k + 1}$ of length $\lfloor \frac{2k + 3}{2} \rfloor = k + 1$.
\end{proof}

For the literal automaton of a finite prefix codes consisting of at least two words, there always exists a word of linear length and logarithmic rank.

\begin{theorem}\label{thm:log-rank}
Let $X$ be a finite prefix code with at least two words.
Let $\mathrsfs{A}_X=(Q,\Sigma,\delta)$ be its partial literal automaton with $n$ states and height $h$.
Then there exists a word of length at most $2h$ and of rank at most $\lceil \log_2 hn \rceil + \lceil \log_2 h \rceil$ for $\mathrsfs{A}_X$.
Moreover, such a word can be found in polynomial time in $n=|Q|$.
\end{theorem}
\begin{proof}
The general idea is as follows.
We construct a word from the theorem in two phases.
First, we define an auxiliary \emph{filtering} algorithm that computes some function $\alpha\colon \Sigma^* \to \Sigma^*$.
We consider the results of the algorithm for a lot of short (logarithmic length) input words $w$ and show that at least one of them satisfies that $\alpha(w)$ is non-mortal, has length at most $h$, and every state from $Q$ is either sent to $\bot$ or goes through the root state of the literal automaton by its action.
Then, we use specific properties of the image $\delta(Q,\alpha(w))$ to divide it into two disjoint sets: one that has up to $h$ states, but on a single specific path, and the other one with a small (logarithmic) number of states. 
In the second phase, we construct a word $v$ of length also bounded by $h$, such that its action maps all the states from the mentioned specific path to a subset of at most logarithmic size.
The concatenation of both words $\alpha(w) v$ is a word of length at most $2h$ satisfying the theorem.

\paragraph*{Selection of pivot.}
Since $X$ consists of at least two words, there exists a state $p$ such that at least two letters have defined the transition from it.
We choose $p$ to be a state with this property that is at the minimal distance from the root state $r$, i.e., such that the length of the shortest word~$w$ such that $\delta(r,w)=p$, is the smallest possible. Let $a, b \in \Sigma$ be two letters with defined transitions going from $p$. 
Since each state on the path from the root to $p$ has only one letter with a defined outgoing transition, such a state $p$ is unique.
For the literal automaton in \Cref{fig:examples} (right), the chosen state~$p$ is~$q_3$.

\paragraph*{Filtering algorithm.} Consider the following auxiliary algorithm that takes a word $w \in \{a, b\}^*$ as the input. Recall that $a, b$ are two letters that have defined transitions going from $p$.
We perform steps for $i=1,2,\ldots$.
In each step, we keep a subset of \emph{active} states $S_i \subseteq Q$, and two words $u_i$, $w_i$. The construction of the algorithm guarantees that the set of active states always remains non-empty and that~$\alpha(w)$ is non-mortal for every word $w \in \{a, b\}^*$.
At the beginning, $S_1 = Q$, $u_1 = \varepsilon$, and $w_1 = w$.
In the $i$-th step, we have two cases.\\
\textit{(Case~1)} If $p$ is active, then let $a$ be the first letter of $w_i$, assuming $w_i$ is non-empty.
We apply letter~$a$ to~$S_i$, and let $S_{i+1}=\delta(S_i,a)$ for the next iteration.
Also, we move the first letter from $w_i$ to the end of~$u_i$, that is, we set $u_{i+1} = u_i a$ and $w_{i+1}$ to be such that $w_i=a w_{i+1}$.\\
\textit{(Case~2)} If $p$ is not active, then let $y$ be the smallest letter, under some fixed order on $\Sigma$ (the same for every step) that is non-mortal for $S_i$, i.e., with $\delta(S_i,y) \neq \emptyset$.
Such a letter always exists in a literal automaton, as for every state we can find a non-mortal letter.
We apply $y$ to $S_i$, so $S_{i+1}=\delta(S_i,y)$, and set $u_{i+1} = u_i y$ and $w_{i+1} = w_i$.\\
\textit{(Termination)} If $w_{i+1}$ is empty or $u_{i+1}$ has length at least $h$, we stop. The word $u_{i+1}$ is the output of the algorithm.

To see that the algorithm terminates, observe that at every step the length of $u_i$ is increased by~$1$, thus there are at most $h$ iterations.

For a word $z$, denote by $\alpha(z)$ the word $u_{i+1}$ obtained in the last step of the algorithm.
Note that $\delta(Q,u_i)=S_i$ holds in every step.
Since every time we apply a letter that is non-mortal for the set of active states, this word is non-mortal for $\mathrsfs{A}_X$.
Its length is also at most $h$.
Observe that for two different words $v_1, v_2$, either the words $\alpha(v_1)$, $\alpha(v_2)$ are different or they have length $h$.
Indeed, if their length is smaller than $h$, then the algorithm terminates when $w_{i+1}$ is empty.
Consider the longest common prefix~$v$ of $v_1$ and $v_2$, so $v_1 = v v'_1$ and $v_2 = v v'_2$.
The algorithm executes in the same way up to the $|v|$-th time when Case~1 is applied.
Then it performs some iterations with Case~2, which do not decrease the length of $w_i$.
Finally, Case~1 must hold, and the algorithm applies the first letters of $v'_1$ and $v'_2$.
At this moment, the two constructed words become different.

\paragraph*{All-through-root word construction.}
Now, we show that there exists a word $w \in \{a,b\}^*$ of length $\lceil \log_2 hn \rceil$ such that for every state $q \in Q$, either the action of some prefix of $\alpha(w)$ maps $q$ to the root $r$, or the action of $\alpha(w)$ is undefined for $q$. Informally, it means that when applying $\alpha(w)$ letter by letter, every state of the literal automaton will be at some point mapped to the root state or killed. Recall also that for every $w \in \{a,b\}^*$, $\alpha(w)$ is non-mortal.

Recall two letters $a,b \in \Sigma$ such that state $p$ has a defined outgoing transition under both of them.
Consider the set of the resulting words $W = \{\alpha(w) \mid w \in \{a,b\}^{\lceil \log_2 hn \rceil}\}$.
Note that the alphabet may contain other letters besides $a$ and $b$, however $W$ is constructed only by applying $\alpha$ to words in $\{a,b\}^*$.
If there is at least one word of length $h$ in $W$, we are done, since $h$ is the height of the literal automaton and~$\alpha(w)$ is guaranteed to be non-mortal for every word $w \in \{a,b\}^*$.
Otherwise, all words in $W$ are of a length less than $h$, and there are at least $hn$ pairwise distinct words.
Then, by the pigeonhole principle, there is a subset $W' \subseteq W$ consisting of at least $n+1$ words of equal length.
Suppose for a contradiction that for every word $\alpha(w) \in W'$, there is a state $q_{\alpha(w)}$ such that $\delta(q_{\alpha(w)}, \alpha(w)) \ne \bot$ and that is not mapped to $r$ by the action of any prefix of $\alpha(w)$.
Note that in the literal automaton of a prefix code, every state except the root state has an incoming transition for exactly one letter.
Thus for a pair $q$, $q'$ of states and two words $t$, $t'$ of the same length, we can have $\delta(q, t) = \delta(q', t) \ne \bot$ only if some prefix of $t$ (or $t'$) has the action mapping $q$ (or $q'$, respectively) to the root state.
Therefore, all states $q_{\alpha(w)}$ for $\alpha(w) \in W'$ are pairwise distinct, and we have at least $n+1$ such states.
This means that the number of states of $\mathrsfs{A}_X$ is larger than $n$, which is a contradiction.
Thus, we get a non-mortal word $\alpha(w)$ of length at most $h$ with the property that every state $q \in Q$ with $\delta(q, \alpha(w)) \ne \bot$ is mapped to the root state by some prefix of~$\alpha(w)$.

\paragraph*{Splitting the image.}
We consider the set $\delta(Q,\alpha(w))$ and split it into two sets, one with specific states and the other one of logarithmic size.
Let $P \subseteq Q$ be the set of all states in the unique shortest path from~$r$ to $p$ excluding $p$ itself.
For the literal automaton in \Cref{fig:examples} (right), $P=\{q_1,q_2\}$.
The size of $P$ is at most $h$.
Note that, by the selection of~$p$, for each state in $P$, exactly one letter has a defined transition.
Thus, the set $Q \setminus P$ has the property that the unique shortest path from~$r$ to each state in $Q \setminus P$ contains~$p$.
Note that $\delta(Q,\alpha(w)) \cap (Q \setminus P)$ has up to $\lceil \log_2 hn \rceil$ states, since by the construction of $\alpha(w)$, at most $|w|$ prefixes of $\alpha(w)$ have the action mapping some state to $p$ (when Case~1 holds).
Below, we deal with the sets $\delta(Q,\alpha(w)) \cap P$ and $\delta(Q,\alpha(w)) \cap (Q \setminus P)$ separately.

\paragraph*{Compressing the part in $P$.}
Denote by $R$ the image $\delta(Q, \alpha(w))$. If $R \cap P=\emptyset$, then we skip this phase, take $v=\varepsilon$ and proceed to the summary.
Otherwise, we show how to construct a word $v$ of length at most $h$ such that $1 \le |\delta(R  \cap P,v)| \le \lceil \log_2 h \rceil$.

Consider the following procedure constructing an auxiliary word $v$.
We perform steps $i=1,2,\ldots$, and in the $i$-th step keep a word $v_i$ and a set $S_i$ of \emph{active} states.
At the beginning, we set $v_1 = \varepsilon$ and $S_1 = R \cap P$.
In the $i$-th step, we do the following.
Let $\ell_i$ be the shortest word mapping some active state from $P$ to $p$.
We apply $\ell_i$ to the set $S$ of active states and let $S'=\delta(S,\ell_i)$.
If $S' \cap P=\emptyset$, we stop.
Otherwise, let $\ell'_i \in \{a,b\}$ be a letter whose transition is undefined for at least half of the states in $S' \cap P$.
At least one of these two letters has this property since each state in $P$ has only one defined transition.
We set $v_{i+1} = v_i \ell_i \ell'_i$ and $S_{i+1} = \delta(S' \setminus \{p\},\ell'_i)$ for the next step; note that we additionally remove $p$ from the set of active states.
If the length of $v_{i+1}$ is at least $h$ or $S_{i+1}=\emptyset$, we stop.
Otherwise, we proceed to the $(i+1)$-th step for $S_{i+1} \subset P$ of size at most $\lfloor |S_i|/2 \rfloor$.

Denote the word $v_{i+1}$ constructed at the last step by $v$.
By the construction, the word $\alpha(w) v$ is non-mortal, because after each step, $p \in \delta(R  \cap P, v_{i+1})$.
The number of active states decreases at least twice in every step, thus the number of performed steps is at most $\lceil \log_2 h \rceil$.
Furthermore, in every step, exactly one active state from $P$ is mapped to a state in $Q \setminus P$, since we use the shortest word such $\ell_i$.
This state is always $p$, which we remove from the active states.
The number of these removed states is at most $\lceil \log_2 h \rceil$ and we know that the other states from $R \cap P$ are mapped to $\bot$ in some step, thus the cardinality of $\delta(R \cap P, v)$ is at most $\lceil \log_2 h \rceil$.

To bound the length of $v$, we observe that the sum of the lengths of words $\ell_i \ell'_i$ from every iteration cannot exceed $h$.
It is because every active state, if it is not mapped to $\bot$, is mapped by the action of $\ell_i \ell'_i$ to a state farther from the root by $|\ell \ell'_i|$.
Thus, the sum of the lengths cannot exceed $h$.

\paragraph*{Summary.}
It follows that the rank of $\alpha(w) v$ is at most $\lceil \log_2 hn \rceil + \lceil \log_2 h \rceil$.
Indeed, the states in~$Q$ that are mapped into the set $P$ by the action of $\alpha(w)$ are then mapped by the action of $v$ to a set of cardinality at most $\lceil \log_2 h \rceil$.
All the other states are mapped by the action of $\alpha(w)$ to a set of size at most $\lceil \log_2 hn \rceil$.
	
Finally, note that the word $\alpha(w)v$ can be found in polynomial time.
In particular, to find $\alpha(w)$, we check $2^{\lceil \log_2 hn \rceil} = \O(hn)$ different words, and every word is processed in polynomial time by the auxiliary algorithm computing $\alpha$.
\end{proof}

\begin{corollary}\label{cor:literal_aut}
Let $\mathrsfs{A}_X$ be a partial literal automaton with $n$ states. If it is synchronizing, its reset threshold is at most $\O(n \log^3 n)$.
If the \v{C}ern\'{y} conjecture holds, then it is at most $\O(n \log^2 n)$.
\end{corollary}
\begin{proof}
If $|X|=1$ then the bound follows from \Cref{pro:code_size_1}.
If $|X|\ge 2$, from \Cref{thm:log-rank}, we get a word $w$ of length $\O(n)$ and rank $\O(\log n)$.
Then we use \Cref{cor:rt_by_induced} with $w$, which yields the upper bound $\O(n) + \O(n) \cdot \rt(\mathrsfs{B})$, where $\mathrsfs{B}$ is an induced automaton with $\O(\log n)$ states.
Then we use upper bounds on the reset threshold of a complete DFA (\cite{Shitov2019,Szykula2018ImprovingTheUpperBound}) transferred to strongly connected partial DFAs by \Cref{thm:cerny_transfer}.
\end{proof}

\section{Lower bounds for properly incomplete DFAs}\label{sec:lower_bounds}

We conclude with observations for transferring lower bounds from the complete case to the partial case.
Of course, in general, this is trivial, since a complete DFA is a special case of a partial DFA.
On the other hand, letters with all transitions undefined cannot be used for synchronization.
Hence we need to add a restriction to exclude these cases and see the effect of usable incomplete transitions.
A partial DFA is \emph{properly incomplete} if there is at least one letter whose transition is defined for some state and is undefined for some other state.

For an a partial DFA $\mathrsfs{A}$, let the length of the shortest non-mortal words of rank at most $r$ be called the \emph{rank threshold} $\rt(\mathrsfs{A},r)$.
We show that bounding the rank/reset threshold of a strongly connected properly incomplete DFA is related to bounding the corresponding threshold of a complete DFA.
A general construction for this is as follows.

\begin{definition}[Duplicating automaton]
For a complete DFA $\mathrsfs{A}=(Q,\Sigma,\delta_\mathrsfs{A})$, we construct the \emph{duplicating automaton} $\mathrsfs{A}^\mathrm{D}=(Q \cup Q',\Sigma \cup \{\gamma\},\delta_{\mathrsfs{A}^\mathrm{D}})$ as follows.
Assume that $Q=\{q_1,\ldots,q_n\}$.
Then $Q'=\{q'_1,\ldots,q'_n\}$ is a set of fresh states disjoint with $Q$ and $\gamma \notin \Sigma$ is a fresh letter.
For all $1 \le i \le n$ and each $a \in \Sigma$, we define:
\[\delta_{\mathrsfs{A}^\mathrm{D}}(q_i,a) = q_i,\quad\delta_{\mathrsfs{A}^\mathrm{D}}(q_i,\gamma) = q'_i,\quad\delta_{\mathrsfs{A}^\mathrm{D}}(q'_i,a) = \delta_\mathrsfs{A}(q_i,a),\quad\text{and }\delta_{\mathrsfs{A}^\mathrm{D}}(q'_i,\gamma) = \bot.\]
\end{definition}

The duplicating automaton turns out to be a partial DFA counterpart to the recent Volkov's construction of complete DFAs \cite{Volkov2019IndempotentLetters}.
The duplicating automaton $\mathrsfs{A}^\mathrm{D}$ has twice the number of states of $\mathrsfs{A}$ and is properly incomplete.
Also, it is strongly connected if $\mathrsfs{A}$ is.

\begin{proposition}\label{pro:lower_bound}
Let $\mathrsfs{A}=(Q,\Sigma,\delta_\mathrsfs{A})$ be a strongly connected complete DFA.
For all $1 \le r < n$, we have $\rt(\mathrsfs{A}^\mathrm{D},r) = 2\rt(\mathrsfs{A},r)$.
\end{proposition}
\begin{proof}
Let $1 \le r < n$, and let $w$ be a shortest word of rank at most $r$ in $\mathrsfs{A}^\mathrm{D}$.
We observe that $w = \gamma a_1 \gamma a_2 \gamma a_3 \dots \gamma a_k$, for some letters $a_i \in \Sigma$.
Indeed, the action of $a_i a_j$ is the same as the action of~$a_i$, thus two consecutive letters from $\Sigma$ cannot occur in a shortest word.
Also, $\gamma^2$ cannot occur, as it is a mortal word.
There is no letter $a_i$ at the beginning of $w$, because $\delta_{\mathrsfs{A}^\mathrm{D}}(Q,a_i \gamma)=\delta_{\mathrsfs{A}^\mathrm{D}}(Q,\gamma)=Q'$.
Finally, $w$ does not contain $\gamma$ at the end, because $\gamma$ at the end cannot change the rank, i.e., for all $S \subseteq Q'$, we have $|\delta_{\mathrsfs{A}^\mathrm{D}}(S,a_i \gamma)| = |\delta_{\mathrsfs{A}^\mathrm{D}}(S,a_i)|$.

We observe that the word $w' = a_1 a_2 \dots a_k$ has rank $r$ in $\mathrsfs{A}$.
Indeed, $\delta_{\mathrsfs{A}^\mathrm{D}}(Q \cup Q',w) \subset Q$, and $q_i \in \delta_{\mathrsfs{A}^\mathrm{D}}(Q \cup Q',w)$ if and only if $q_i \in \delta_\mathrsfs{A}(Q,w')$.
Since $w$ is a word of rank $r$ in $\mathrsfs{A}^\mathrm{D}$, so is $w'$ in $\mathrsfs{A}$.
Thus $\rt(\mathrsfs{A}^\mathrm{D},r)=|w|=2|w'| \ge 2\rt(\mathrsfs{A},r)$.

Conversely, having a shortest word $w' = a_1 a_2 \dots a_k$ of rank $r$ in $\mathrsfs{A}$, we can construct\\${w = \gamma a_1 \gamma a_2 \gamma a_3 \dots \gamma a_k}$ of rank $r$ in $\mathrsfs{A}^\mathrm{D}$, thus $\rt(\mathrsfs{A}^\mathrm{D},r) \le 2\rt(\mathrsfs{A},r)$.
\end{proof}

From  \Cref{pro:lower_bound}, it follows that we cannot expect an upper bound on the reset threshold of a properly incomplete strongly connected DFA better than $0.04135 n^3 + \O(n^2)$, unless we can improve the best general upper bound on the reset threshold of a complete DFA, which currently is roughly $0.1654 n^3 + \O(n^2)$ \cite{Shitov2019}.
We can also show a lower bound on the largest possible reset threshold, using the Rystsov's construction of DFAs with long shortest mortal words \cite{Rystsov1997ResetWordsForCummutativeAndSolvableAutomata}.

\begin{proposition}\label{pro:properly_incomplete_extremal}
For every $n$, there exists a strongly connected properly incomplete $n$-state DFA with the reset threshold $\frac{n^2 - n}{2}$.
\end{proposition}
\begin{proof}
Let $\mathrsfs{A} = (Q, \Sigma, \delta)$ be a strongly connected partial DFA with only one state $r$ having undefined outgoing transitions for some letter, and let $w$ be its shortest mortal word. Let $w = w'a$, where~$a \in \Sigma$.
Then $\delta(Q, w') = \{r\}$.
Moreover, since the DFA is strongly connected, every prefix $w''$ of $w'$ of length smaller than $|w|-(|Q| - 1)$ has the action mapping $Q$ to a set of size at least two; otherwise, there is a shorter word mapping $Q$ to $\{r\}$.
Thus we get that $\mathrsfs{A}$ is synchronizing, and the length of its shortest reset word is at least $|w| - (|Q| - 1) - 1 = |w| - |Q|$.
For every $n$, Rystsov \cite{Rystsov1997ResetWordsForCummutativeAndSolvableAutomata} constructed a strongly connected properly incomplete $n$-state DFA with the only one state having undefined outgoing transitions. The length of its shortest mortal word is equal to $\frac{n^2 + n}{2}$. Thus we get a lower bound of $\frac{n^2 + n}{2} - n$ on the length of its shortest reset word. It can be seen directly from the construction of the DFA that this bound is tight.
\end{proof}

\acknowledgements{We thank the anonymous reviewers for their comments that improved the presentation of the paper.}

\bibliographystyle{plainurl}
\bibliography{synchronization}
\end{document}